\documentclass{article}
\usepackage[utf8]{inputenc}
\usepackage{amsmath}
\usepackage{amssymb}
\usepackage{graphicx}
\usepackage{amsthm}
\newtheorem{theorem}{Theorem}
\newtheorem{corollary}{Corollary}

\newtheorem{prop}{Proposition}
\newtheorem{lemma}{Lemma}
\newtheorem{example}{Example}

\newtheorem{remark}{Remark}
\usepackage{algpseudocode}
\usepackage{algorithmicx}
\usepackage{algorithm}

\newcommand{\beq}{\begin{equation}}
\newcommand{\eeq}{\end{equation}}
\newcommand{\barr}{\left[\begin{array}}
\newcommand{\earr}{\end{array}\right]}

\newcommand{\bpf}{\begin{proof}}
\newcommand{\epf}{\end{proof}}

\newcommand{\ftwo}{\ensuremath{\mathbb{F}_{2}}}

\newcommand{\ff}{\ensuremath{\mathbb{F}}}

\newcommand{\GOE}{\mbox{\rm GOE}\,}
\newcommand{\COI}{\mbox{\rm COI}\,}
\title{Deciding One to One property of Boolean maps: Condition and algorithm in terms of implicants.}
\author{Virendra Sule\\Professor of Practice\\Department of Computer Science and Engineering\\
Indian Institute of Technology Hyderabad\\ India\\virendra.sule@prjt.cse.iith.in}
\date{September 9, 2025}
\begin{document}
\maketitle
\begin{abstract}
This paper addresses the computational problem of deciding invertibility (or one to one-ness) of a Boolean map $F$ in $n$-Boolean variables. This problem is a special case of deciding invertibilty of a map $F:\mathbb{F}_{q}^n\rightarrow\mathbb{F}_{q}^n$ over the finite field $\mathbb{F}_q$ for $q=2$. Algebraic condition for invertibility of $F$ is well known to be equivalent to invertibility of the Koopman operator of $F$ as shown in \cite{RamSule}. In this paper a condition for invertibility is derived in the special case of Boolean maps $F:B_0^n\rightarrow B_0^n$ where $B_0$ is the two element Boolean algebra in terms of \emph{implicants} of Boolean equations defined by the map. This condition is then extended to the case of general maps in $n$ variables and $m\geq n$ equations. Hence this condition answers the special case of invertibility of maps $F$ defined over the binary field $\mathbb{F}_2$ alternatively, in terms of implicants instead of the Koopman operator. The problem of deciding invertibility of a map $F$ (or that of finding its Garden of Eden (GOE)) over finite fields is distinct from the satisfiability problem (SAT) or the problem of deciding consistency of polynomial equations over finite fields. Hence the well known algorithms for deciding SAT or of solvability using Grobner basis for checking membership in an ideal generated by polynomials is not known to answer the question of invertibility of a map. Similarly it appears that algorithms for satisfiability or polynomial solvability are not useful for computation of GOE of $F$ even for maps over the binary field $\mathbb{F}_2$. 
\end{abstract}

\noindent
\begin{tabular}{ll}
\emph{Arxiv Subject categories}\/: & cs.DM, cs.CC, cs.LO, math.AG\\
\emph{ACM classification}\/: & I.1.2, G.2.0\\
\emph{Keywords}\/: & Polynomial maps, Boolean equations,\\
 & Finite Fields, Implicants of Boolean functions,\\ 
 & Permutation polynomials.
\end{tabular}
\section{Introduction}
Mappings in finite domains are omnipresent in applications of Computational Sciences. Due to the complex connections of causes (inputs) and effects (outputs) the problems such as deciding invertibility and computing an inverse of a map at a given local point in the image turn out to be $NP$ hard. It appears that even for an approach using Artificial Intelligence (AI) to solving such a problem, efficient computation of intermediate mathematical problems can be enormously useful. For concreteness, consider $F$ to represent a non-linear map in the domain which is a cartesian product of a finite field. Proving $F$ to be invertible is a fundamental mathematical problem and whose computational complxity depends on not only the size of its domain but also on whether $F$ is linear or non-linear. Apart from the linear maps $F$ in vector spaces there is no easy condition known for invertibility of general non-linear maps which frequently arise in applications. Hence it is not only of theoretical interest to pose invertibilty of maps but is of high practical importance to pose it as a computational question and understand the complexity of inversion. An objective of this paper is to explore the computational side of the invertibilty question and to provide an answer for the special case of Boolean maps.

Consider first the algebraic version of the invertibility problem. Let $\ff$ be a finite field and $n$ a fixed number. Given a polynomial map $F:\ff^n\rightarrow\ff^m$, defined by $m$ polynomials in $n\leq m$ variables (indeterminates),
\beq\label{PolyMap}
F=(f_1(X_1,\ldots,X_n),\ldots,f_m(X_1,\ldots,X_n))
\eeq
For $m=n$, $F$ is called \emph{locally invertible} at a given point $y$ in $\ff^m$ if there exists a unique $x$ in $\ff^n$ such that $F(x)=y$, or called \emph{invertible} if there exist solutions $x$ to $F(x)=y$ for every $y$ and $F(x)=F(\tilde{x})$ implies $x=\tilde{x}$. When $m>n$, we call $F$ \emph{locally one to one} if there exists a unique solution $x$ for a given $y$ and \emph{one to one} when for every $y$ in the image of $F$, $F(x)=F(\tilde{x})$ implies $x=\tilde{x}$. These decision problems are required to be solved by describing an algorithm which accepts $F$ as an input. The complexity of the algorithm is described in terms of $n$ for a fixed field $\ff$. As a general condition for invertibility when $m=n$, it follows that $F$ is invertible iff all orbits of iterations $F^{(k)}(x)$ for arbitrary $x$ in $\ff^n$ are closed. In other words, $F$ is a permutation of $\ff^n$. Such an algorithm based on iterates of $F$ to find whether every orbit is closed, clearly has exponential complexity if at least one of the orbits has an exponential number of points since the cardinality $|\ff^n|=p^n$ is exponential for characteristic $p$. Another problem close to invertibility of a map is that of proving uniqueness of the solution, i.e. proving that there exists a unique $x$ in $\ff^n$ which is a solution of the polynomial equations
\beq\label{EquationSystem}
f_i(X_1,\ldots,X_n)=0,i=1,\ldots,m
\eeq
for $m\geq n$. Even if the algebraic form of the map $F$ as $m$ polynomials over $\ff$ in $n\leq m$ variables is given, existence of unique solutions of equations for $F(x)=y$ for a given $y$ in $\ff^m$ is not decidable by any known algorithm deciding solvability of these equations. Hence constructing an algorithmic solution to deciding invertibilty of $F$ or uniqueness of solutions $x$ of equations $F(x)=y$ for a given $y$, is a well defined yet poorly explored problem over finite fields. (In fact, invertibility of polynomial maps and uniqueness of solutions of polynomial equations are among the unresolved problems in Algebraic Geometry which surface in terms of unsolved problems such as the Jacobian problem). Alternatively, map $F$ over a finite field (when $n=m$) is invertible iff the set of points in the $\GOE$ of $F$, the complement of the image of $F$, is empty. In this paper we shall present an algorithm for the special case of the problem of deciding the invertibility of $F$ over $\ff_2$ for $m=n$ or deciding the uniqueness of solutions of equations $F(x)=y$ over $\ftwo$ for a given $y$ in the general case when $m\geq n$ using implicants of Boolean equations. Alternatively, we also develop an algorithm to decide invertibility as well as formulating equations whose solutions form $\GOE$ of $F$ and its generalisation for $m>n$. In the case $m>n$ the analogue of $\GOE$ is called as the \emph{Complement of Image} $\COI$ of $F$.

In this paper invertibility of square maps (when $m=n$) and uniqueness of solutions $x$ for more general maps (when $m>n$) is presented in terms of implicants of Boolean equations $F(x)=y$. Since computing a complete set of implicants of Boolean equations is equivalent to representing all solutions of a system of Boolean equations, the problem of deciding invertibility or one to one-ness is a harder problem than satisfiability (SAT) problem since the SAT problem is concerned with only deciding existence or otherwise of a solution of such a system.

\subsection{Previous work}
One of the papers which considered the question of invertibility of polynomial maps over real fields is \cite{Newman}. The field being real, there are different issues which need to be addressed as compared to when the field is finite. The paper \cite{Newman} mainly proves that if a polynomial map in two variables over reals is one to one then it is also onto. However the case of several variables ($n>2$) is not addressed in the paper. Over finite fields $\ff$ this question does not occur for $F:\ff^n\rightarrow \ff^n$ because when the domain is finite one to one and onto are equivalent. In \cite{RamSule} authors showed that over finite fields, above map $F$ is one to one iff the reduced Koopman operator of the map $F$ in the space of $\ff$ valued functions over $\ff^n$ is an invertible linear map. An explicit inverse of $F$ is also constructed in the paper using a basis of an invariant subspace of the vector space of polynomial functions over $\ff$. The size of the Reduced Koopman operator is however exponentially large in general in $n$ but for certain low degree (quadratic) maps $F$ there are special cases of maps (defined by feedback shift registers) for which the dimension of the Koopman operator is of polynomial size. Hence in such special cases the invertibility of $F$ can be answered in terms of the non-singularity of the reduced Koopman operator in polynomial time. Hence the algorithm based on Koopman operator is one of the non-brute force algorithms but still has exponential complexity in general. In finite fields all maps are polynomial maps. Those which are in one indeterminate and invertible (called permutation polynomials) have been studied for long \cite{ZhengWangWei}. The conditions for a polynomial to be a permutation in the finite field of its definition in terms of its co-efficients are well known \cite{RamSuleInverse}. However these do not resolve the problem of deciding the invertibility in many variables $n>1$. In fact all the past references on permutation polynomials, for instance those referred in \cite{ZhengWangWei} are for one variable polynomial maps in finite fields. An explicit formula for inverse of multivariable maps and permutation polynomials in terms of the Koopman operator is developed in \cite{RamSuleInverse}. Grobner basis based algorithms for solving polynomial systems over finite fields \cite{VonZerGathen} are also not known to have addressed the invertibility question of maps. This is possibly due to the fact that the question of invertibility is distinct from the statement of Hilbert's nullstellentzas. Boolean functions and equations have a long history of developments since George Boole's treatise ``Laws of Thought (1954)" \cite{Boole}. While Boolean functions were of primary interest in Logic, Shannon gave insight into the application of Boolean arithmetic for design of switching networks \cite{Shannon}. The question of synthesis of invertible Boolean maps did not explicitly get addressed in Shannon's theory of switching circuits. In the 1960s Boolean arithmetic became important for the satisfiability theory \cite{HBSAT, SchonToran} and the algorithms for CNF-SAT. Again, SAT theory addresses a decision problem which is different from invertibility of Boolean maps or uniqueness of solutions. Boolean functions and logic of Boolean equations has been exhaustively treated by Rudeanu \cite{Rude} and Brown \cite{Brown}. Unfortunately these works by \cite{Rude,Brown} have somehow been poorly appreciated in the Computer Science literature. A recent survey of Boolean elimination theory described in \cite{Rude,Brown} and its applications to computation have been documented extensively in \cite{CramaHammer}. Several equivalent conditions for invertibility of Boolean maps have been stated and proved in \cite{Rude}. All these conditions are mathematical equivalences of invertibilty of Boolean maps. However these conditions further require algorithms for deciding the invertibility of computation of GOE or COI without which they cannot be concrete from a computational point of view. Finally, it must be clarified that the condition for invertibility in terms of Koopman operator of $F$ is not directly applicable to the problem of one to one-ness of a more general map $F:\ff^n\rightarrow\ff^m$ when $m>n$. This paper is based on the previous article \cite{Sule25} by the author.

\section{An implicant based condition for invertibility using the graph of $F$}
As described in the appendix the set of all solutions of a Boolean system of equations can be represented by a complete set of OG implicants of the system. The algorithm described in \cite{Sule2017} is an efficient parallel non-brute force algorithm for computation of a complete set of OG implicants of a Boolean system. We now show that the invertibilty of a Boolean map $F$ can be described in terms of the factors of the implicants, which correspond with points $(x,y)$ in the graph of the map $F$.

Given the map $F:\ftwo^n\rightarrow\ftwo^n$ by an $n$-tuple of polynomial functions as in (\ref{PolyMap}), the \emph{graph} of the map in $\ftwo^n\times\ftwo^n$ is the set of all $2n$-tuples over $\ftwo$, $(x_1,\ldots,x_n,y_1,\ldots,y_n)$ denoted as $(x,y)$ such that $F(x)=y$ i.e. are all solutions $(x,y)$ of the system of Boolean equations
\beq\label{SystemXY}
f_i(X)=Y_i, i=1,2,\ldots,n
\eeq
in indeterminates $X=(X_1,\ldots,X_n)$ and $Y=(Y_1,\ldots,Y_n)$. An implicant of the system (\ref{SystemXY}) is a term $t$ which is a product of literals $X_i$ or $X_i'$ and $Y_i$ or $Y_i'$ such as $t(X,Y)=t_1(X)t_2(Y)$ such that solutions $X,Y$ of $t(X,Y)=1$ satisfy (\ref{SystemXY}). Hence every such implicant denotes a set of points on the graph of $F$. A \emph{complete} set of implicants of (\ref{SystemXY}) is a set of implicants $I(F)$ such that for every point $(x,y)$ in the graph there exists $t$ in $I(F)$ such that $t(x,y)=1$. Let $I(F)$ be a complete set of implicants of the Boolean system (\ref{SystemXY})). Then the graph of $F$ is represented by the solutions $(x,y)$ of the equation
\beq\label{Eqnofgraph}
\sum_{t\in I(F)}t(X,Y)=1
\eeq
while the set of all points in the graph of $F$ is the union
\beq\label{Graph}
\bigcup_{t\in I(F)}S(t)
\eeq
where $S(t)=\{(a,b),t(a,b)=1\}$ is the set of all satisfying assignments of $t$. Let each such implicant be factored as $t(X,Y)=r(X)s(Y)$ in literals involving only $X$ and $Y$. In the following we consider the set of implicants $I(F)$ to be complete and also \emph{orthogonal} (OG) (which means $t_it_j=0$ for $i\neq j$). For a complete set of orthogonal implicants $I(F)$ the set of all points on the graph of $F$ is the disjoint union
\[
\bigsqcup_{i\in I(F)}S(t)
\]

\begin{theorem}\label{Thinvertibility}
\emph{Let $I(F)$ be a complete and OG set of implicants of the system (\ref{SystemXY}). The map $F$ is invertible iff any of the following equivalent conditions are satisfied.
\begin{enumerate}
 \item $\sum_{t_i\in I(F)}s_i(Y)=1$
 \item $GOE(F)=\emptyset$
\end{enumerate}
}
\end{theorem}

\begin{proof}
    Necessity of 1): Since the set of implicants $I(F)$ is complete for the system (\ref{SystemXY}), a point $(x,y)$ in $\ftwo^n\times\ftwo^n$ is in the graph of the map $F$ or in the solution set of the system (\ref{SystemXY}) iff $t_i(x,y)=1$ for at least one $i$ and further since $I(F)$ is an OG set, there is a unique such $i$. Hence $r_i(x)s_i(y)=1$. Since the map is one to one it is also onto hence the set of all $y$ in $\ftwo^n$ appear in solutions of $s_i(Y)=1$ over all implicants $t_i$ in $T$. Hence $\sum_{t_i\in I(F)}s_i=1$ identically. 
    
    Sufficiency of 1): If 1) holds identically then every point $y$ in $\ftwo^n$ is in the projection of the graph of $F$ or the solution set of the system (\ref{SystemXY}). This shows $F$ is onto hence invertible. 
    
    Condition 1) also implies that every point in $\ftwo^n$ is in the image of $F$ which implies $GOE(F)=\emptyset$. Conversely, if $GOE(F)$ is empty then every point $y$ appears as a component of a point in the graph of the map $F$. Hence $\sum_{t_i\in I(F)}s_i(y)=1$ for all $y$, that is the condition 1) is satisfied identically. Hence the two conditions are equivalent. We note that the orthogonality of the implicants is not related to invertibility but is useful to get a unique set of implicants to represent the graph of $F$. 
\end{proof} 

\begin{remark}
    \emph{Invetibility is harder than SAT}\/: Above theorem states the condition for invertibility of $F$ in terms of implicants of the system of equations $F(X)=Y$. Hence the problem of deciding invertibility is equivalent to representing all solutions of a system of Boolean equations. Consequently, the problem of deciding invertibility is harder than deciding consistency of equation which is equivalent to the SAT problem. 
\end{remark}

\subsection{Representation of $GOE(F)$}
The graph of $F$ is the set of points described by the disjoint union (\ref{Graph}) of satisfying assignments of orthogonal implicants $t_i=r_is_i$ of the system of equations (\ref{SystemXY}). The set $T$ of implicants is a complete OG set of implicants of the equations (\ref{SystemXY}). From this it follows that the set of points $Y$ in $\ftwo^n$ which do not belong to the image of $F$ are precisely those for which $t_i(x,Y)=0$ for any $x$ satisfying $r_i(X)=1$. Hence the set $GOEF(F)$ is characterised as
\[
    GOE(F)=\{Y\in\ftwo^n|s_i(Y)=0\forall i\}
\]
Hence we get

\begin{theorem}\label{ThGOE}
\emph{$GOE(F)$ is the set of all satisfying assignments $y$ of the Boolean equation
\beq\label{GOEeqn}
\sum_{t_i\in T}s_i(Y)=0
\eeq
}
\end{theorem}

Hence we also get an alternative characterization of the invertibility of the map as

\begin{corollary}
\emph{Let $I(F)$ be a complete OG set of implicants $t_i=r_is_i$ of the system (\ref{SystemXY}), then the map $F$ is one to one iff the equation 
\[
\sum_{t_i\in I(F)}s_i(Y)=0
\]
is inconsistent (or unsatisfiable).
}
\end{corollary}

Clearly the statement of the corollary is same as saying that 
\[
\sum_{t_i\in I(F)}s_i(Y)=1
\]
identically, which is the necessary and sufficient condition for invertibility of $F$.

\subsection{Examples}
\begin{example}
\emph{Consider the map
\[
    F=
    \barr{l}
    x_1x_3\\ x_2x_3\\ x_1x_4\\ x_2x_4\oplus 1
    \earr
\]
The $Y$ components of an OG implicant set of equations $F(X)=Y$ are given by sets
\[
\begin{array}{l}
\{y_1y_2y_3y_4', y_1'y_2'y_3y_4',y_1'y_2'y_3'y_4\}\\
\{y_1y_2'y_3y_4,y_1y_2'y_3'y_4,y_1'y_2'y_3y_4,y_1'y_2'y_3'y_4\}\\
\{y_1'\}\times\{\mbox{all minterms in } y_2,y_3,y_4\}
\end{array}
\]
Hence the $Y$-components of the implicants are minterms in $Y$ variables but form a subset of the set of all minterms. Since sum of all minterms in $Y$ variables equals $1$ the sum of a partial set of minterms in $Y$ cannot equal $1$. Hence the map is not one to one. Following minterms in $Y$ are missing from the set of $Y$ components of implicants above.
\[
    \begin{array}{l}
    y_1'y_2'y_3'y_4',y_1'y_2y_3'y_4,y_1y_2'y_3y_4',y_1y_2y_3'y_4',\\
    y_1y_2y_3'y_4,y_1y_2y_3y_4,y_1y_2'y_3'y_4'
    \end{array}
    \]
Hence $GOE(F)$ is the set of all $Y$ assignments satisfied by these minterms.
}
\end{example}

\begin{example}
\emph{Consider the FSR map $\Phi:(x_1,x_2,x_3)\mapsto(x_2,x_3,x_1\oplus x_2x_3)$. This is obviously one to one by the known Golomb's condition for a non-singular FSR. The equations (\ref{SystemXY}) for this map are
\[
x_2\oplus y_1=0,x_3\oplus y_2=0,x_1\oplus x_2x_3\oplus y_3=0
\]
To find the implicants of this system we choose the ON system $\{x_2,x_2'x_3,x_2'x_3'\}$ and compute implicants by substituting each of the terms.
\begin{enumerate}
\item $t=x_2'x_3'$. $y_1=0$,$y_2=0$,$x_1\oplus y_3=0$. Hence implicants are 
\[
x_2'x_3'y_1'y_2'y_3'x_1', x_2'x_3'y_1'y_2'y_3x_1
\]
The $s(Y)$ components of these implicants are $y_1'y_2'y_3', y_1'y_2'y_3$.
\item $t=x_2'x_3$. $y_1=0$, $y_2=1$, $x_1\oplus y_3=0$. Hence components of implicants in $Y$ are, $y_1'y_2y_3$, $y_1'y_2y_3'$.
\item $t=x_2$. $y_1=1$, $x_3\oplus y_2=0$, $x_1\oplus x_3\oplus y_3=0$. The $Y$ components of implicants are $y_1y_2(y_2\oplus y_3), y_1y_2'(y_2\oplus y_3),y_1y_2(y_2\oplus y_3\oplus 1), y_1y_2'(y_2\oplus y_3\oplus 1)$. This is the set
\[
\{y_1y_2y_3',y_1y_2'y_3,y_1y_2y_3,y_1y_2'y_3'\}
\]
\end{enumerate}
Hence $s(Y)$ factors of a complete set of implicants is the set
\[
\{y_1'y_2'y_3', y_1'y_2'y_3, y_1'y_2y_3, y_1'y_2y_3', y_1y_2y_3', y_1y_2'y_3, y_1y_2y_3,
y_1y_2'y_3'\}
\]
Hence all the eight minterms in $Y$ variables are present whose sum is always $1$. This verifies that the map $\Phi$ is onto since there is no point in the $GOE(\Phi)$}.
\end{example}

\subsection{Computational condition for permutation polynomials in $\ff_{2^n}$}
Consider a polynomial $f(X)$ in one variable with co-efficients in the binary extension field $\ff_{2^n}$
\[
f(X)=\sum_{i=0}^{n-1}a_iX^i
\]
for $a_i$ in $\ff_{2^n}$. Consider a basis $\{b_0,b_1,\ldots,b_n\}$ of $\ff_{2^n}$ over $\ftwo$. Then for any evaluation at $X=x$ in $\ff_{2^n}$ and expansion of $a_i$ in the basis there are expressions
\[
\begin{array}{lcl}
X & = & \sum_{i=0}^{n-1}x_ib_i\\
a_i & = & \sum_{j=0}^{n}c_{ij}b_j\mbox{ for }i=1,\ldots,n
\end{array}
\]
Substitution and organizing as coefficients of basis terms, $f$ has the expression
\[
f=\sum_{i=1}^{n}f_i(x_1,\ldots,x_n)b_i
\]
Hence the co-ordinate functions of $F$ are computed using a choice of a basis to define
\[
F=(f_1(x_1,\ldots,x_n),\ldots,f_n(x_1,\ldots,x_n))
\]
Once these are computed the question whether $f$ is a permutation polynomial is determined by $F$. 

\begin{prop}
    \emph{$f$ is a permutation polynomial in $\ff_{2^n}$ iff $F:\ftwo^n\rightarrow\ftwo^n$ is invertible.
    }
\end{prop}

\subsection{Condition for one to one-ness of a more general map}
Theorem (\ref{Thinvertibility}) gives the necessary and sufficient condition for invertibility of the map $F:\ftwo^n\mapsto\ftwo^m$ for $m=n$. When $m\neq n$, two cases can occur. For $m<n$ there must be at least two distinct $n$-tuple assignments $x$, $\tilde{x}$ such that $F(x)=F(\tilde{x})$. Hence $F$ cannot be one to one. Hence we assume that $m\geq n$. Then the system of equations (\ref{SystemXY}) is of the form
\beq\label{SystemXYm}
f_i(X_1,\ldots,X_n)=Y_i,i=1,2,\ldots,m
\eeq
abbreviated as the system $F(X)=Y$. Let $T=\{t_i(X,Y)\}$ denote a complete set of OG implicants of this system. Each $t_i$ has factor terms as $t_i(X,Y)=r_i(X)s_i(Y)$. Then we have the following theorem.

\begin{theorem}\label{Thmgeqn}
    \emph{$F$ is one to one iff the set of all terms $\{s_i(Y)\}$ for $t_i$ in $T$ is a distinct set of minterms of cardinality $|\{s_i\}|=2^n$.
    }
\end{theorem}

\begin{proof}
    \emph{Necessity}\/: Consider the graph of the system (\ref{SystemXYm}) which is the set 
    \[
    G=\cup_{t_i\in T}\{(x,y)|t_i(x,y)=r_i(x)s_i(y)=1\}
    \]
    Due to $T$ being an OG set, for each point $(x,y)$ in $G$ there is a unique $i$ such that $t_i(x,y)=1$. Which implies $r_i(x)=s_i(y)=1$. Since $F$ is one to one, no two points such as $(x,y)$, $(\tilde{x},y)$ exist in $G$ with the same $y$ component. Hence there are no free variables in $r_i$ as well as in $s_i$ factors of any implicant $t_i$. Hence since the set $T$ is OG, each of the sets $\{r_i\}$ and $\{s_i\}$ are distinct minterms. Further every point $x$ in $\ftwo^n$ is a unique satisfying assignment of $r_i(X)$ for a unique $i$. Hence cardinality of $\{s_i\}$ equals $2^n$. 
    
    \emph{Sufficiency}\/: If $\{s_i\}$ is a set of minterms of cardinality $2^n$, then
    \[
    \sum_{t_i\in T}s_i(Y)=1
    \]
    has a solution set of points $y$ in $\ftwo^m$ of cardinality $2^n$. (This is the image set of $F$ in $\ftwo^m$). For each of these points $y$ in $\ftwo^m$, there is a unique $x$ in $\ftwo^n$ such that $(x,y)$ is in $G$ the graph of $F$. Hence $F$ is one to one.
\end{proof}

For $m=n$ the assertion that $s_i$ are $2^n$ distinct minterms in $y$ implies that the sum
\[
\sum_{t_i\in T}s_i(Y)=1
\]
which is the result of Theorem \ref{Thinvertibility}.

We can have an analogous characterization of points $y$ in $\ftwo^m$ for which points $(x,y)$ do not belong to the set of solutions (or the graph) of equations (\ref{SystemXYm}) for maps when $m>n$. These cannot be termed as $\GOE$ of the map $F$ since a recursive dynamical trajectory cannot be defined with these points as initial conditions. Hence we shall call this set of points as the \emph{Complement of Image} of $F$, denoted $\COI(F)$

\begin{corollary}
    A point $y$ in $\ftwo^m$ is in $\COI$ of $F$ iff $y$ does not satisfy $s_i(Y)=1$ for one of the factors $s_i$ of an implicant $t_i$ of (\ref{SystemXYm})
\end{corollary}

\begin{proof}
    By the previous theorem image of $F$ is the union of the set of all satisfying assignments of $\{s_i(Y)=1|t_i\in T\}$. Equivalently, $y$ is in the $\COI$ iff the equations $F(X)=y$ are not consistent.
\end{proof}

\subsection{Condition for unique solution of a Boolean system}
As a special case of the problem of deciding one to one ness of a more general map considered above, we now discuss the problem of deciding unique solution of a system of Boolean equations over the binary field or the two element Boolean algebra $B_0$. The problem we want to address now is to give a condition for the existence of a unique solution of the system of equations
\beq\label{BoolSys}
f_i(X_1,\ldots,X_n)=0\mbox{ for }i=1,\ldots,m
\eeq
Where $f_i$ are Boolean functions mapping $B_0^n$ to $B_0$ for the two element Boolean ring $\{0,1\}$ with usual addition $\oplus$, multiplication $.$ and complement $'$. The solution of the system is an assignment in $B_0$ of the indeterminates $X_i$ such that the equations are satisfied. In this special case we get the following statement of the condition for uniqueness of solutions of (\ref{BoolSys}).

\begin{theorem}
    Let $T$ denote an OG and complete set of implicants of the system (\ref{BoolSys}). Then the system has a unique solution iff $T$ is a set consisting of a single minterm in all variables of the system.
\end{theorem}

\begin{proof}
    By the property of completeness of the set of implicants $T$ of the Boolean system (\ref{BoolSys}), all solutions of the system are represented by
    \[
    S=\cup_{t_i\in T}\{x\in B_0^n|t_i(x)=1\}
    \]
    Since $t_i$ are also orthogonal, the individual satisfying assignments 
    \[
    S(i)=\{x|t_i(x)=1\}
    \]
    of implicants are distinct. Hence $S$ consists of a unique solution implies that each of $S$ is a single point. Hence there is a single implicant $t_i$ whose satisfying assignment is a single point in the solution set $S$ i.e. $t_i$ does not have a free variable assignments for any of its variables. Hence $T$ is a set of a single minterm in all the variables of the system.
\end{proof}

Above theorem gives the condition for uniqueness of a solution of a Boolean system in terms of implicants of the system. Hence it is necessary to have an efficient algorithm for computing the implicants of Boolean systems. Efforts in developing such algorithms are described in the next section.

\section{Implicant based algorithms}
For a system of Boolean equations in $n$ Boolean variables $X_i$ such as  
\beq\label{Boolsys}
f_i(X_1,\ldots,X_n)=0,i=1,\ldots,m
\eeq
an implicant $t$ is a term in literals of variables $X_i$ which satisfies $t(x)=1$ whenever $x=(x_1,\ldots,x_n)$ is a satisfying assignment of (\ref{Boolsys}). A set of implicants $I$ is called a complete set of implicants of (\ref{Boolsys}) if for every satisfying assignment $x$ there is an implicant $t$ in $I$ such that $t(x)=1$. Alternatively, an implicant of a Boolean function given as a product of factors
\beq\label{Prodfun}
H=\prod_{i=1}^{m}h_i(X_1,\ldots,X_n)
\eeq
is a term $t$ which satisfies $t(x)=1$ whenever $h_i(x)=1$ for all $i$. A set of implicants of $I$ of $H$ is called complete if for every assignment $H(x)=1$ there is a $t$ in $I$ such that $t(x)=1$. A parallel algorithm for computing a complete OG set of implicants of a Boolean system of $m$-equations in $n$-variables or that of functions $H$ given as products of Boolean functions was announced in \cite{Sule2017}. Both of these problems can be mapped into eachother by taking complements $h_i=f_i'$. A complete set of OG implicants compactly (i.e. efficiently in terms of space required) represents all satisfying assignments of the Boolean system (\ref{Boolsys}) or the function (\ref{Prodfun}). Computation of such a representation of all satisfying assignments is far more efficient than brute force enumeration of all satisfying assignments due to the fact that OG implicants represent a very large number of assignments in their free variables (for which arbitrary assignments can be made). Hence a very large sets of satisfying assignments can be represented by a small set of OG implicants. Hence deciding invertibility of a map $F$ or computation of the $GOE(F)$ by using the implicant based parallel algorithm turns out to be an efficient scalable non-brute force methodology.

\subsection{Computation of an Implicant set by decomposition}
A system of Boolean equations such as (\ref{BoolSys}) or a complex Boolean function such as (\ref{Prodfun}) are practically often available in terms of functions $f_i$ (and $h_i$) which are a sparse (i.e. have a very small number of variables compared to the total number of variables in the system or $H$). Also due to the sparseness there exist clusters (subsets) of equations or factors $h_i$ with possibly no overlapping variables. Hence implicant computations for these clusters can be carried out independently or in parallel. This is the decomposition which makes the implicant computation efficient. There are many different ways of decomposition which we shall not address in this paper. Different decomposition heuristics can result in efficient computation of complete sets of OG implicants as shown in a practical case study of solving a cryptanalysis problem in practically feasible time \cite{SPACE2019}. 

\begin{remark}
It is worth considering the computational achievement reported in \cite{SPACE2019} in relation to brute force key search in Cryptanalysis. For instance the only case study of successful brute force search in the past was for the 56 bit key of the DES algorithm. It is well known that 64 bit brute force search is beyond the capability of current supercomputers in realistic time. While the cryptanalysis in \cite{SPACE2019} using the implicant based solver results in a key search of an 80 bit cipher Bivium in a feasible time of 49 hours and 5 terabytes of memory. A brute force search for 80 bit unknowns otherwise would have been infeasible.   
\end{remark}

\subsection{Algorithm for computation of a complete OG set of implicants}
We now present an algorithm for parallel decomposition of the Boolean system (\ref{Boolsys}) or a product (\ref{Prodfun}) for computation of a complete and OG set of implicants. Consider a Boolean function $f$ and let $S(f)$ denote its set of satisfying assignments (i.e. $f(x)=1$ iff $x$ belongs to $S$). If $I(f)$ is a complete set of OG implicants of $f$, then we have
\[
S(f)=\bigsqcup_{t\in I(f)}S(t)
\]
the above union is disjoint because the implicants are OG (hence for $t,s\in I(f)$, $ts=0$ which is equivalent to $S(t)\cup S(s)=\emptyset$). An important operation of ratio of a Boolean function $f$ with a term $t$ is defined as the Boolean function denoted
\[
f/t=f(t(x)=1)
\]

\subsubsection{Description of $I(fg)$}
Let $I(f)$ and $I(g)$ be complete sets of OG implicants of $f$ and $g$ respectively. Then we have
\begin{lemma}
\[
\begin{array}{lcl}
I(fg) & = & \bigsqcup_{t\in I(f)}\{ts,s\in I(g/t)\}\\
 & = & \bigsqcup_{s\in I(g)}\{ts,t\in I(f/s)\}
\end{array}
\]    
\end{lemma}
As a special case then it follows that if $f$ and $g$ have non-overlapping variables then
\[
I(fg)=I(f)\times I(g)=\{ts,t\in I(f),s\in I(g)\}
\]
This formula shows that if $f$, $g$ have non overlapping variables then a complete OG set of implicants can be computed in parallel by independent computations of $I(f)$ and $I(g)$ and then forming the product set $I(fg)$. Hence we have the

\begin{theorem}\label{Th:decomposition}
Let $H$ be described in terms of factors
\[
H=(\prod_{k=1}^{N}f_k)H_1
\]
where $f_k$ are factors with mutually non-overlapping variables between $f_k$ and $f_l$ for $k\neq l$, then 
\[
I(H)=\{t_1t_2\ldots t_Ns,t_i\in I(f_i)s\in H_1/\prod_{i} t_i\}
\]
\end{theorem}

\begin{proof}
    By induction. Let $N=2$ and $H=fgh$ where $f$ and $g$ have non-overlapping variables. Then by the above lemma and the special expression of $I(fg)$ it follows that
    \[
    I(H)=\{tsr,t\in I(f),s\in I(g),r\in I(H/ts)\}
    \]
    Hence the theorem holds for $N=2$.

    Assuming the induction step for $N-1$, define
    \[
    f=\prod_{i=1}^{N-1}f_i
    \]
    so that $H=ff_kH_1$. $f$ and $f_k$ have non-overlapping variables hence by the $N=2$ case
    \beq\label{recursiveformula}
    I(H)=\{tsr,t\in I(f),s\in I(f_k),r\in I(H_1/ts)\}
    \eeq
    Hence the theorem is true for $N$ since
    \[
    t=\prod_{i=1}^{N-1}t_i
    \]
\end{proof}
The formula (\ref{recursiveformula}) is recursive and at each step decomposition and computation of implicants is feasible in parallel. Following algorithm is developed to make use of this recursive structure and parallel decomposition in computation of the OG implicant set of a Boolean function $H$.

\subsubsection{Algorithm for computation of $I(H)$ using parallel decomposition}
We now describe the algorithm for computation of implicants of a function $H$ with a large sparse factors such as in (\ref{Prodfun}) by parallel computation. First we assume that a function 
\[
\mbox{\textsc{Implforsimple}}(f)
\]
is available for computation of a complete OG set of implicants $I(f)$ of a function $f$ in small number of variables (defined with a bound $m$). For instance, this can be done by evaluating the minterms $\mu$ in the variables of $f$ and choosing those with $f/\mu=1$ as the OG implicants. Other minters $\mu$ then satisfy $f/\mu=0$. A parallel Algorithm for this function is proposed in \cite{Sule2017} and has parallel time complexity of $O(n)$ for $n$ the number of variables in $f$. The algorithm is described as Algorithm \ref{Decomposition} below. Although the complexity of the algorithm \textsc{Implforsimple} is exponential (i.e. $O(n)$ for large $n$, the function is expected to be called for sufficiently small $n$ during independent (parallel) computations in the following algorithm at the cost of space complexity.

\begin{algorithm}\label{Decomposition}
\caption{Algorithm for computing a complete OG set of implicants of a function by decomposition in terms of non-overlapping variables}
\begin{algorithmic}[1]
\Procedure{Implicant}{Compute a complete OG set of implicants of $H$}
\State \textbf{Input}
\State 1. $H$ given in terms of its sparse factors $h_i$ as in (\ref{Prodfun})
\State 2. The bound $m$ for the function $\mbox{\textsc{Implforsimple}}(f)$
\If {number of variables in $H$ is $\leq m$} 
\State Return 
    \[
    I(H)=\mbox{\textsc{Impforsimple}}(H)
    \]
    \Else
    \Repeat
    \State Select a clusters of factors $\{f_k\}$ with non-overlapping variables 
    \State such that
    \[
    H=H_1\prod_{k}f_k
    \]
    where $H_1$ may have overlapping variables with some of the $f_k$ such that each $f_k$ has at most $m$ variables.
    \State Compute complete sets of OG implicants $T_k$ of each $f_k$ \State independently (parallely) by calling the function 
    \[
    T_k=\mbox{\textsc{Implforsimple}}(f_k)
    \]
    \State  
    \For {$t=\prod {T_k}$}
    \State \% Recursive call
    \[
    I(H)=\bigcup_{s\in\mbox{\textsc{Implicant}}(H_1/t)}\{ts\}
    \]
    \EndFor
    \State $H\leftarrow H_1/t$
    \Until Number of variables in $H$ is $\leq m$.  
    \State Return $I(H)$ \% $(I(H)$ is a complete OF set of implicants of $H$.
\EndIf
\EndProcedure
\end{algorithmic}
\end{algorithm}

The algorithm clearly shows parallel computation of implicants by decomposing equations into clusters of equations with non-overlapping variables. Hence the algorithm works much better than brute force search. Although the complexity is still exponential in the number of variables, the algorithm facilitates efficient decomposition of large systems with sparse functions or equations. Other alternative heuristics than finding clusters of equations with non-overlapping variables also exist for efficient decomposition of equations to find the implicant set. These are reported in \cite{SPACE2019}. Such heuristics can be used for efficient decomposition and improving scalability of computation of implicants.

\section{One to one-ness in argument variables}
In the previous section we described one to one-ness of a map $F:\ftwo^n\rightarrow\ftwo^m$ in terms of the implicants of the equations $F(X)=Y$ of the graph of $F$. For sparse maps $F$ the decomposition of $F$ in compositions of maps $F_1\circ F_2$ described in Section \ref{Decomposition} can be utilized before applying Theorem \ref{Thmgeqn} to decide one to one-ness of each of the maps $F_i$ in $F=F_1\circ F_2$. Hence such a sparse decomposition can make the computation of deciding one to one-ness of a complex map into parallel computations.
In this section we consider another approach to describe one to one-ness of the map $F$ which is in terms of implicants in argument variables $X_i$ and their associates $\tilde{X}_i$ as described in the following lemma.

\begin{lemma}\label{Leminxandxtilde}
   The map $F:\ftwo^n\rightarrow\ftwo^m$, $m\geq n$ in $n$-variables $X_i$ is one to one iff the set of all solutions $S_1=\{(x,\tilde{x})\}$ in $\ftwo^n\times\ftwo^n$ of the equation
   \beq\label{fxftildex}
   F(x)=F(\tilde{x})
   \eeq
   belong to the diagonal 
   \[
   D=\{(x,\tilde{x})|x=\tilde{x}\}
   \]
\end{lemma}

\begin{proof}
    $F$ is one to one implies that all solutions of equation (\ref{fxftildex}) are $x_i=\tilde{x}_i$ which are in the diagonal $D$. If the points in the diagonal $D$ are the only solutions of (\ref{fxftildex}), then $F$ is one to one.
\end{proof}

To utilize the observation of the above lemma consider a compact notation for minterms, in literals of $X_i$,
\[
m(X,\bar{a})=(X_1)^{a_1}(X_2)^{a_2}\ldots(X_n)^{a_n}
\]
where $\bar{a}=(a_1,\ldots,a_n)$ in $\ftwo^n$, the literal $(X_i)^{a}=X_i$ for $a=1$ and $(X_i)^{a}=X_i'$ for $a=0$.

\begin{theorem}\label{Thonxandxtilde}
Map $F:\ftwo^n\rightarrow\ftwo^m$, $m\geq n$ is one to one iff the set of $2^n$ elements
\beq\label{Impinxandxtilde}
M=\{m(X,\bar{a})m(\tilde{X},\bar{a})\forall\bar{a}\in\ftwo^n\}
\eeq
is the set of all OG implicants of the Boolean system (\ref{fxftildex}).    
\end{theorem}

\begin{proof}
    If the set $M$ of minterms in literals $X$ and $\tilde{X}$ is the complete set of implicants of equations (\ref{fxftildex}) then all satisfying assignments $x_i$, $\tilde{x}_i$ of
    \[
    m(X,\bar{a})m(\tilde{X},\bar{a})=1
    \]
    is equal to the set of all solutions of (\ref{fxftildex}). Since this set of assignments is in the diagonal $D$, $F$ is one to one by the above lemma. 

    Conversely, if $F$ is one to one for each assignment $x$ in $\ftwo^n$ there exists $\tilde{x}=x$ in $\ftwo^n$ such that the evaluation $m(x,\bar{a})m(\tilde{x},\bar{a})=1$. Hence the set $M$ is a complete set of OG implicants of (\ref{fxftildex}).
\end{proof}

Theorem \ref{Thonxandxtilde} cannot be used for verifying one to one-ness of a given map $F$ in practice due to the exponential size of the set $M$.  

\appendix
\section{Boolean functions, equations and implicants}
We present here a very brief overview of essential theory of Boolean functions, equations and implicants to provide background for the previous sections. Reader is referred to \cite{Rude, Brown} for the basic theory of Boolean functions, equations and logic.

\subsection{Implicants of Boolean functions and equations}
 The binary field $\ftwo$ is the Boolean ring corresponding to the two element Boolean algebra $B_0$. Since this is the most relevant field from the point of view of applications, we shall throughout use the field $\ftwo$ for defining Boolean maps. Although many results are true for Boolean functions over an arbitrary finite Boolean algebra $B$, we shall consider their counterparts for the special case of $B_0$ algebra. A Boolean function in $n$ variables, $f(x_1,\ldots,x_n)$ is a map $f:\ftwo\rightarrow\ftwo$ defined by its truth table of values for all strings $\{0,1\}^n$ as arguments. The set of all Boolean functions in $n$ variables themselves form a Boolean algebra denoted $B_0(n)$ under formal rules of Boolean addition, multiplication and complement in variables. The Boole-Shannon formula of partial evaluation of Boolean functions shows that a Boolean function has a unique \emph{minterm} canonical form also called (SOP) form, a unique \emph{maxterm} canonical form, also called (POS) form and further a unique \emph{algebraic normal form} (ANF) which is a polynomial representation of the Boolean function. A map $F:\ftwo^n\rightarrow\ftwo^n$ is an $n$-tuple of Boolean functions 
\beq\label{MapFpolynomials}
\left[f_1(X_1,\ldots,X_n),\ldots,f_n(X_1,\ldots,X_n)\right]
\eeq
 where $f_i$ are polynomials in $n$ variables with co-effcients in $\ftwo$. A set of Boolean functions $\Phi=\{\phi_1,\ldots,\phi_m\}$ in $B_0(n)$ is said to be \emph{orthogonal} (OG) if $\phi_i\phi_j=0$ for $i\neq j$ and \emph{orthonormal} (ON) if it is OG and satisfies
\[
\sum_{i=1}^{m}\phi_i=1
\]
A fundamental fact about Boolean functions is that given an ON set $\Phi$ any Boolean functions $f$ in $B_0(n)$ has a representation
\beq\label{ONexp}
f=\sum_{\phi_i\in\Phi}\alpha_i\phi_i
\eeq
where functions $\alpha_i$ belong to the interval $[f\phi_i,f+\phi_i']$. This formula gives rise to the representation of all satisfying assignments of a Boolean system and an efficient algorithm for such a representation. The reader is referred to \cite{Sule2017} for a detailed discussion of this algorithm. The above OG representation formula is a generalization of the minterm canonical form SOP of a function (and dually can be written to generalise the maxterm representation as POS form).

A notion of \emph{ratios of Boolean functions} can be defined for special Boolean functions. If $f(X)$ and $g(X)$ are Boolean functions in $n$-variables denoted as $X$ then the ratio $f/g$ is defined for assignments $X=x$ when $g(x)=1$ by the value $f(x)$ i.e.\ $f/g=f(x)|(g(x)=1)$. An \emph{implicant} of a Boolean function $f$ is a term $t=l_{i_1}l_{i_2}\ldots l_{i_k}$ where $l_{i_k}$ are literals $x_{i_k}$ or $x_{i_k}'$ for variables $1\leq i_k\leq n$ such that whenever $t=1$ for any arguments $x_i=a_i$, then $f(a_i)=1$. This is denoted as $t\leq f$. If $f(X),g(X)$ are Boolean functions and $f(X)=g(X)$ is an equation denoted $E$ then $t$ is said to be an implicant of the equation $E$ if whenever $t(a_i)=1$ for some assignments $x_i=a_i$ then the equation is satisfied at $a_i$, $f(a_i)=g(a_i)$. A set of implicants $I$ is said to be a \emph{complete} set of implicants for a function $f$ or an equation $E$ if whenever for assignments $a_i$, $f(a_i)=1$ or $E$ is satisfied at $a_i$, there exists a $t$ in $I$ such that $t(a_i)=1$. A set of implicants $I=\{t_1,t_2,\ldots,t_m\}$ is said to be \emph{orthogonal} (OG) if $t_it_j=0$ for $i\neq j$. It follows that if $I$ is a complete OG set of implicants of $f$ then
\beq\label{Implicantexp}
f=\sum_{i=1}^{m}t_i
\eeq
and all solutions to the equation $f=1$ denoted $S(f)$ (the set of satisfying assignments of $f$, are given by the disjoint union
\beq\label{Allsolns}
S(f)=\bigsqcup_{i=1}^{m}S(t_i)
\eeq
where $S(t)$ denotes all assignments $x_i=a_i$ such that $t(a_i)=1$. Instead, if an equation $E$ is given a complete set of implicants represents all solutions of the equation $E$. For a system of Boolean equations and a set of Boolean functions an implicant $t$ is considered as a simultaneous implicant of all equations and functions. The set of all satisfying assignments of a system of equations or the product of functions can be described by a complete set of OG implicants as above for a single equation or a function. An algorithm to compute such a complete set of OG implicants is described in \cite{Sule2017} which is a non-brute force algorithm to represent the set of all satisfying assignments of a system of equations or a product of functions. Representation of all solutions of a system of Boolean equations or a product of functions by a complete OG set of implicants is very compact and efficient in storage because the number of implicants is much smaller than the number of solutions. 

\subsection{$GOE(F)$}
Let $F:\ftwo^n\rightarrow\ftwo^n$ be a Boolean map. Then the set $GOE(F)$ is the set of all points $y$ in $\ftwo^n$ such that $F(x)=y$ has no solution. Hence $GOE(F)$ is the complement of the image of $F$. It is also characterised as the set of points $x$ such that the iterates $F^{(k)}(x)$, $k=1,2,\ldots$ never return to $x$ or alternatively the set of points which are not periodic points of the map $F$. Hence it follows that $F$ is one to one iff $GOE(F)=\emptyset$. Using the implicant based representation of all solutions we can construct a non-brute force algorithm to represent the set $GOE(F)$. Hence such an algorithm gives a solution to the problem of deciding invertibility of $F$.
\section{Conclusion}
The well known Grobner basis algorithm as well SAT algorithms are not useful for deciding invertibility (or one to one ness) of maps in finite fields. This paper proposes an algorithm to solve this problem in the special case when the field is binary or the map is Boolean in terms of implicants of Boolean systems. A complete set of implicants of a Boolean system represents all solutions of the system hence the problem of deciding invertibility is harder than the SAT problem. The algorithm is parallel and non-brute force and can be much efficient if the map is defined by sparse functions. The time complexity for parallel implicant computation is $O(n)$ where $n$ is the number of variables. Space complexity however is exponential in the number of variables and is dependent on the properties of the map. The algorithm is based on a previously known parallel algorithm to characterize implicants of a Boolean system of equations. The algorithm is scalable for large number of variables if there exist possibilities of decomposition due to non overlapping variables as in sparse maps which arise naturally in practice. The more general finite field case of the problem is not solvable by this algorithm hence is open for further research. Unlike in the previous work on explicit formula of the inverse of a square map in terms of the reduced Koopman operator, results of this paper do not give such an explicit inversion formula. The central takeaway of this research is that an algorithmic solution for the problem of deciding one to one ness of Boolean systems and maps is proposed which can give a scalable solution for the practical case of sparse maps using parallel computation. Further improvements can lead to the practical use of this solution for intermediate computations in AI based computing paradigms.
%\printbibliography
%\bibliographystyle{abbrv}

\end{document}